\pdfoutput=1
\documentclass[american,aps,pra,reprint, superscriptaddress]{revtex4-1}
\usepackage[unicode=true,pdfusetitle, bookmarks=true,bookmarksnumbered=false,bookmarksopen=false, breaklinks=false,pdfborder={0 0 0},backref=false,colorlinks=false] {hyperref}
\hypersetup{ colorlinks,linkcolor=myurlcolor,citecolor=myurlcolor,urlcolor=myurlcolor}
\usepackage{graphics,epstopdf,graphicx, amsthm, amsmath, amssymb, times, braket, color, bm}
\usepackage[up]{subfigure}
\usepackage{cleveref}
\definecolor{myurlcolor}{rgb}{0,0,0.7}

\newcommand{\tinyspace}{\mspace{1mu}}

\newcommand{\proj}[1]{| #1\rangle\!\langle #1 |}

\newcommand{\Prok}[2]{| #1\rangle\!\langle #2 |}

\newcommand{\iinner}[2]{\langle #1 | #2\rangle}

\DeclareMathOperator{\trace}{Tr}
\newcommand{\Ptr}[2]{\trace_{#1}\Pa{#2}}
\newcommand{\Tr}[1]{\Ptr{}{#1}}
\newcommand{\Innerm}[3]{\left\langle #1 \left| #2 \right| #3 \right\rangle}
\newcommand{\Pa}[1]{\left[#1\right]}

\newcommand{\abs}[1]{\left\lvert\tinyspace #1 \tinyspace\right\rvert}
\newcommand{\norm}[1]{\left\lVert #1 \right\rVert}

\theoremstyle{plain}
\newtheorem{thm}{Theorem}
\newtheorem{lem}[thm]{Lemma}
\newtheorem{prop}[thm]{Proposition}
\newtheorem{cor}[thm]{Corollary}

\newcommand*{\myproofname}{Proof}
\newenvironment{mproof}[1][\myproofname]{\begin{proof}[#1]}{\end{proof}}
\def\ot{\otimes}
\def\complex{\mathbb{C}}
\def\real{\mathbb{R}}

\def\cI{\mathcal{I}}

\def\cD{\mathcal{D}}

\def\cH{\mathcal{H}}

\def\cB{\mathcal{B}}

\makeatother

\begin{document}

 \author{Kaifeng Bu}
 \email{kfbu@fas.harvard.edu}
 \affiliation{School of Mathematical Sciences, Zhejiang University, Hangzhou 310027, PR~China}
 \affiliation{Department of Physics, Harvard University, Cambridge, MA 02138, USA}

  \author{Lu Li}
 \affiliation{School of Mathematical Sciences, Zhejiang University, Hangzhou 310027, PR~China}
    \author{Shao-Ming Fei}
    \email{feishm@cnu.edu.cn}
 \affiliation{School of Mathematical Sciences, Capital Normal University, Beijing 100048, PR China}
  \affiliation{Max-Planck-Institute for Mathematics in the Sciences, 04103 Leipzig, Germany}

 \author{Junde Wu}
 \affiliation{School of Mathematical Sciences, Zhejiang University, Hangzhou 310027, PR~China}

\title{Distribution of coherence in bipartite systems based on incoherent-quantum coherence measures}

\begin{abstract}
The distribution of quantum coherence in multipartite systems is one of the basic problems in the resource theory of coherence. While the usual coherence measures are defined on a single system and  cannot capture the nonlocal correlation
between subsystems, in order to  deal with the distribution of coherence it is crucial to quantify the coherence in bipartite systems properly. Here, we introduce incoherent-quantum (IQ) coherence measures on  bipartite systems, which can characterize the correlations between systems. According to the  IQ
coherence measures on bipartite systems, we find the distribution of coherence of formation and assistance  in bipartite systems: the
 total coherence of formation  is lower  bounded by the sum of
 coherence of formation  in each subsystem and the entanglement of formation between the subsystems,
 while the total coherence of assistance is upper bounded by the sum of
 coherence of assistance in each subsystem and the entanglement of assistance between subsystems.
 Besides, we also obtain the tradeoff relation between
 the coherence cost  and entanglement cost,  distillable coherence and distillable
 entanglement in bipartite systems.
Thus, the IQ coherence measures introduced here
truly capture the nonlocal correlation between subsystems and  reveal the distribution of coherence in bipartite systems.
\end{abstract}

\maketitle

\section{Introduction}

Quantum coherence, stemming from the superposition rule of quantum mechanics, can capture the feature of quantumness in a single system,
and  play an important role in
 a variety of applications
ranging from   thermodynamics \cite{Lostaglio2015,Lostaglio2015NC} to metrology \cite{Giovannetti2011}.
Recently, following the method in quantum information theory, the resource theory of coherence has been developed \cite{Baumgratz2014,Girolami2014,Streltsov2015,Winter2016,Killoran2016,Chitambar2016,Chitambar2016a}.
Besides quantum coherence, there are other resource theory including quantum entanglement \cite{HorodeckiRMP09}, asymmetry \cite{Bartlett2007,Gour2008,Gour2009,Marvian2012,Marvian2013,Marvian14,Marvian2014}, thermodynamics \cite{Fernando2013}, and steering \cite{Rodrigo2015}, where all these quantum resource are helpful to quantum information processing tasks.

Any resource theory consists of two basic elements:
free states and free operation. The state (operation)
outside the sets of free states (operation) is called
resource. For example, the free states in the resource theory of coherence is called incoherent states
and the corresponding free operations is called incoherent operations \cite{Baumgratz2014}.
The resource measures are introduced to
quantify the amount of  resource in a given quantum state.
To quantify the coherence in  a single system, several
operational coherence measures has been proposed, namely,
relative entropy of coherence \cite{Baumgratz2014}, $l_1$ norm of coherence \cite{Baumgratz2014}, coherence of formation \cite{Winter2016},
robustness of coherence \cite{Napoli2016}, coherence weight \cite{Bu2017asym}  and max-relative entropy of coherence \cite{Bu2017b,Chitambar2016b},  where relative entropy of coherence characterizes the optimal rate to distill maximally coherent state
from a given quantum state \cite{Winter2016}, coherence of formation is equal to minimal cost of maximal coherent state to prepare
the given state \cite{Winter2016} and max-relative entropy of coherence can be interpreted as the maximal overlap with the maximally coherent state
under incoherent operations \cite{Bu2017b}.

However, as the observation and characterization of the properties of quantum systems is often affected by the
coupling to the environment,
the effect of environment on quantifying coherence has to be taken into account.  Thus, incoherent-quantum (IQ) coherence measures on bipartite systems are introduced here, which can not only quantify the coherence in local subsystem but also the collective coherence between systems and thus
plays a crucial role in the distribution of coherence in multipartite systems \cite{Bu2017c}.
In view of the significance of
IQ coherence measures, we investigate the properties of IQ coherence measures in details and  the distribution of coherence in bipartite systems in terms of other coherence measures
such as coherence of formation and assistance.

Here, we introduce  incoherent-quantum (IQ) coherence measures defined by
relative entropy, max-relative entropy and $l_1$ norm on bipartite systems to quantify the
coherence in the system with the access to a quantum memory.
We also introduce the IQ coherence of formation and assistance on bipartite systems, by which we find the distribution of coherence formation and  assistance  in bipartite systems: the total coherence of formation is lower  bounded by the sum of coherence of formation  in each  subsystem and
entanglement of formation  between subsystems,
while the total coherence of assistance is  upper bounded by the sum of coherence of assistance in each  subsystem and
entanglement of  assistance between subsystems.
Besides, we  find the relationship between coherence cost ( distillable coherence ) and entanglement cost ( distillable entanglement) in bipartite systems.
Moreover, we obtain the monogamy relationship for IQ coherence measures ( such as IQ coherence of assistance and formation ) in tripartite systems, which illustrates the distribution of coherence
in multipartite systems.
Furthermore, we discuss the relationship between different IQ coherence measures, such as the equivalence between IQ coherence measures
defined by max-relative entropy and $l_1$ norm.
\section{Preliminaries}\label{sec:pre}
Let $\cH$ be a d-dimensional Hilbert space and $\cD(\cH)$ be the set
of density operators acting on $\cH$. Let us first recall some
basic facts about max- and min- relative entropies and the resource theory of coherence.

\noindent
{\it \textbf{Max-and min-relative entropy}.}---Given two operators $\rho$ and $\sigma$
with $\rho\geq0$, $\Tr{\rho}\leq 1$ and $\sigma\geq0$, the max-relative entropy of $\rho$ with respect to $\sigma$ \cite{Datta2009IEEE,Datta2009} is defined as
\begin{eqnarray*}
D_{\max}(\rho||\sigma):=\min\set{\lambda\in\real_+:\rho\leq2^\lambda\sigma},
\end{eqnarray*}
where $D_{\max}(\rho||\sigma)$ is well defined if $\text{supp}[\rho]\subset \text{supp}[\sigma]$ with $\text{supp}[\rho]$ being the
support of $\rho$.
The min-relative entropy of $\rho$ with respect to $\sigma$ \cite{Datta2009IEEE,Datta2009} is defined as
\begin{eqnarray*}
D_{\min}(\rho||\sigma):=-\log\Tr{\Pi_{\rho}\sigma},
\end{eqnarray*}
where $\Pi_{\rho}$ is the projector on the support of $\rho$.

\noindent
{\it \textbf{Resource theories of quantum coherence}.}---Given a fixed  reference basis $\set{\ket{i}}^{d-1}_{i=0}$
for some d-dimensional Hilbert space, any quantum
state which is diagonal in the reference basis is the free state
in the resource theory of coherence and the set of incoherent states is denoted by $\cI$. However, there is still general
consensus on the set of free operations in the resource theory of coherence. Here, we refer
incoherent operations (IO) \cite{Baumgratz2014} as the free operations, where
incoherent operations (IO) is the set of all  quantum operations $\Phi$ that
admit a set of Kraus operators $\set{K_i}_i$ such that $\Phi(\cdot)=\sum_iK_i(\cdot)K^\dag_i$ and
$K_i\cI K^\dag_i\subset\cI$ for any $i$ \cite{Baumgratz2014}. Besides,
several operational coherence measures have been proposed, which are listed as follows,

(i) $l_1$ norm of coherence \cite{Baumgratz2014},
\begin{eqnarray*}
C_{l_1}(\rho) = \sum_{\substack{{i,j=0}\\{i\neq j}}}^{d-1} |\bra{i}\rho\ket{j}|,
\end{eqnarray*}

(ii)  relative entropy of coherence \cite{Baumgratz2014},
\begin{eqnarray*}
C_r(\rho)=S(\Delta(\rho))-S(\rho),
\end{eqnarray*}
where $S(\rho)=-\Tr{\rho\log\rho}$ is von Neumann entropy,

(iii)  max-relative entropy of coherence \cite{Bu2017b},
\begin{eqnarray*}
C_{\max}(\rho)=\min_{\sigma\in\cI}D_{\max}(\rho||\sigma),
\end{eqnarray*}

(iv) coherence of formation \cite{Winter2016},
\begin{eqnarray*}
C_f(\rho)=\min_{\rho=\sum_ip_i\proj{\psi_i}}\sum_ip_iS(\Delta(\psi_i)),
\end{eqnarray*}
where the minimization is taken over all pure state decomposition of
$\rho$,

(v) coherence of assistance \cite{Chitambar2016},
\begin{eqnarray*}
C_a(\rho)=\max_{\rho=\sum_ip_i\proj{\psi_i}}\sum_ip_iS(\Delta(\psi_i)),
\end{eqnarray*}
where the maximization is taken over all pure state decomposition of $\rho$,

(vi) coherence weight \cite{Bu2017asym},
\begin{eqnarray*}
C_w(\rho)&=& \min\left\{s\geq0:\rho=(1-s)\sigma+s\tau,  \sigma\in \mathcal{I}, \right.\nonumber\\
&&\left.\tau\in\cD(\cH)\right\}.
\end{eqnarray*}

\section{ entropic IQ coherence measure }
Given a bipartite system $\cH_A\ot\cH_B$ with a fixed basis $\set{\ket{i}_A}_i$ of $\cH_A$, we can define the
 relative entropy of incoherent-quantum (IQ) coherence for any bipartite state $\rho_{AB}\in\cD(\cH_A\ot\cH_B)$ as follows  \cite{Chitambar2016},
\begin{eqnarray*}
C^{A|B}_r(\rho_{AB})=\min_{\sigma_{A|B}\in\cI Q}S(\rho_{AB}||\sigma_{A|B}),
\end{eqnarray*}
where  the set of incoherent-quantum  states $\cI Q$ \cite{Chitambar2016,StreltsovPRX2017} is given by
\begin{eqnarray*}\nonumber
\nonumber\cI Q=\{\sigma_{A|B}\in\cD(\cH_A\ot \cH_B)|\sigma_{A|B}=\sum_ip_i\sigma^A_i\ot \tau^B_i,\\ \nonumber
\sigma^A_i~ \text{is incoherent}, \tau^B_i\in\cD(\cH_B)\},
\end{eqnarray*}
and $C^{A|B}_r$ gives an upper bound for assistant distillation of coherence \cite{Chitambar2016,StreltsovPRX2017}.
Max- and min-relative entropies of IQ  coherence have also been defined  in  Ref. \cite{Bu2017c} as follows,
\begin{eqnarray*}
C^{A|B}_{\max}(\rho_{AB})
=\min_{\sigma_{A|B}\in \cI Q}D_{\max}(\rho_{AB}||\sigma_{A|B}),\\
C^{A|B}_{\min}(\rho_{AB})
=\min_{\sigma_{A|B}\in \cI Q}D_{\min}(\rho_{AB}||\sigma_{A|B}),
\end{eqnarray*}
where $C^{A|B}_{\max}$ captures the maximal advantage  of bipartite states in certain subchannel discrimination problems \cite{Bu2017c}.

For IQ coherence measure $C^{A|B}$, the following properties are considered:
(i) positivity, $C^{A|B}(\rho_{AB})\geq 0$ and $C^{A|B}(\rho_{AB})\geq 0$
iff $\rho_{A|B}\in \cI Q$;
(ii) monotonicity under  incoherent operation  on A side, that is, $C^{A|B}(\Lambda^A_{IO}\ot \mathbb{I}(\rho_{AB}))\leq C^{A|B}(\rho_{AB})$;
(iii) strong monotonicity under incoherent operation on A side, that is, for incoherent operation   $\Lambda^A_{IO}(\cdot)=\sum_i K^A_i(\cdot)K^{A\dag}_i$
with $K^A_i\cI K^{A\dag}_i\subset\cI$,
$\sum_i p_iC^{A|B}(\tilde{\rho}_i)\leq C^{A|B}(\rho)$,
where $p_i=\Tr{K^A_i\rho_{AB} K^{A\dag}_i}$ and $\tilde{\rho}_i=K^A_i\rho_{AB} K^{A\dag}_i/p_i$;
(iv) monotonicity under quantum operation  on B side, that is, $C^{A|B}(\mathbb{I}\ot \Lambda^B(\rho_{AB}))\leq C^{A|B}(\rho_{AB})$;
(v) convexity, that is,  for $\rho_{AB}=\sum^n_ip_i\rho_i$,
$C^{A|B}(\rho_{AB})\leq\sum_{i}C^{A|B}(\rho_i)$.

Note that,
 $C^{A|B}_r$ satisfies all these properties, and $C^{A|B}_{\max}$ satisfies all these properties except (v).
 However, $C^{A|B}_{\max}$ satisfies the quasi-convexity instead of convexity,
that is,  for $\rho_{AB}=\sum^n_ip_i\rho_i$,
$C^{A|B}_{\max}(\rho_{AB})\leq\max_{i}C^{A|B}_{\max}(\rho_i)$ ( See \cite{Bu2017c}).
For bipartite pure state $\ket{\psi}_{AB}$, it  can be written as $\ket{\psi}_{AB}=\sum^{d_A}_{i=1}\sqrt{p_i}\ket{i}_A\ket{u_i}_B$ in the local basis $\set{\ket{i}_A}_i$ of $\cH_A$, thus
$C^{A|B}_r(\psi_{AB})=-\sum_ip_i\log p_i=S(\Delta_A(\rho_A))\leq \log d_A$ and $C^{A|B}_{\max}(\psi_{AB})=2\log(\sum^{d_A}_{i=1}\sqrt{p_i})\leq\log d_A$.
Due to the convexity of $C^{A|B}_r$ and quasi-convexity of $C^{A|B}_{\max}$, $C^{A|B}_r(\rho_{AB})\leq \log d_A$ and $C^{A|B}_{\max}(\rho_{AB})\leq \log d_A$ for
any bipartite state $\rho_{AB}\in\cD(\cH_A\ot\cH_B)$, which means that the maximal value for IQ coherence measures $C^{A|B}_r$ and $C^{A|B}_{\max}$
does not depend on the subsystem B.
Here, we only consider the monotonicity of IQ coherence measures under
local incoherent operations on A side and local quantum operations on B side, while
 the (strong) monotonicity of IQ coherence under  the local
incoherent-quantum operations and classical communication (LIQCC) \cite{Chitambar2016,StreltsovPRX2017} is still unknown as the
characterization of LIQCC (such as the Kraus operators of LIQCC) is unclear.

According to the definition, $C^{A|B}_r(\rho_{AB})\geq C_r(\rho_A)$ with
 $\rho_A=\Ptr{B}{\rho_{AB}}$ being the reduced state.
For any pure bipartite  state $\ket{\psi}_{AB}$, the following relation holds,
\begin{eqnarray}\label{prop:pure_cr}
C^{A|B}_r(\psi_{AB})=C_r(\rho_A)+S(\rho_B),
\end{eqnarray}
where $S(\rho_B)$ is the von Neumman entropy of the reduced state $\rho_B=\Ptr{A}{\proj{\psi}_{AB}}$ on system B.
This comes directly from the definition of $C^{A|B}_r$ and the fact that $S(\rho_A)=S(\rho_B)$ for
pure bipartite state.
In general, for any bipartite state $\rho_{AB}$, $C^{A|B}_r(\rho_{AB})\geq C_r(\rho_A)+\delta_{A\to B} $ \cite{Bu2017c}, where $\delta_{A\to B}$ is the quantum discord between A and B for state $\rho_{AB}$ \cite{Ollivier2001}.
For pure tripartite  states, we have the following proposition.

\begin{prop}\label{prop:diff_tri_pu}
Given a pure tripartite state $\ket{\psi}_{ABC}\in\cD(\cH_A\ot\cH_B\ot\cH_C)$, it holds that
\begin{eqnarray}
C^{A|B}_r(\rho_{AB})-C^{A|C}_r(\rho_{AC})=S(\rho_B)-S(\rho_C),
\end{eqnarray}
where $\rho_{AB}, \rho_{AC}$ are the corresponding reduced states of $\psi_{ABC}$.

\end{prop}
\begin{proof}
Any tripartite pure state $\ket{\psi}_{ABC}$ can be written as
$\ket{\psi}_{ABC}=\sum_i\sqrt{p_i}\ket{i}_A\ot\ket{u_i}_{BC}$ in the local basis $\set{\ket{i}_A}_i$ of $\cH_A$ with $p_i\geq0$, $\sum_ip_i=1$, and thus
the reduced states $\rho_{AB}$ and $\rho_{AC}$ can be written as
\begin{eqnarray*}
\rho_{AB}=\sum_{i,j}\sqrt{p_ip_j}\Prok{i}{j}_A\ot \Ptr{C}{\Prok{u_i}{u_j}_{BC}},\\
\rho_{AC}=\sum_{i,j}\sqrt{p_ip_j}\Prok{i}{j}_A\ot \Ptr{B}{\Prok{u_i}{u_j}_{BC}}.
\end{eqnarray*}
Due to the definition,
\begin{eqnarray*}
C^{A|B}_r(\rho_{AB})&=&S(\Delta_A(\rho_{AB}))-S(\rho_{AB})\\
&=&S(\Delta_A(\rho_A))+\sum_ip_iS( \Ptr{C}{\proj{u_i}_{BC}})\\
&&-S(\rho_{AB}),\\
C^{A|C}_r(\rho_{AC})&=&S(\Delta_A(\rho_{AC}))-S(\rho_{AC})\\
&=&S(\Delta_A(\rho_A))+\sum_ip_iS( \Ptr{B}{\proj{u_i}_{BC}})\\
&&-S(\rho_{AC}).
\end{eqnarray*}
As for pure states, the von Neumann entropy of the reduced states is equal,
thus $S( \Ptr{C}{\proj{u_i}_{BC}})=S(\Ptr{B}{\proj{u_i}_{BC}})$ for any i,
and $S(\rho_{AB})=S(\rho_C)$, $S(\rho_{AC})=S(\rho_B)$. Therefore, we obtain the
result.

\end{proof}
The Proposition \ref{prop:diff_tri_pu} illustrates that  the difference between $C^{A|B}_r(\rho_{AB})$ and $C^{A|C}_r(\rho_{AC})$ for tripartite pure state is equal to
the difference between the amount of information encoded in  ancillary systems B and C.

In tripartite systems, the monogamy relation for  relative  entropy of IQ coherence has been proposed as
$C^{AB|C}_r(\rho_{ABC})\geq C^{A|C}_r(\rho_{AC})+C^{B|C}_r(\rho_{BC})$ \cite{Bu2017c}, where $\rho_{AB}$ and  $\rho_{AC}$ are the
corresponding reduced states.
However, the relationship between $C^{A|BC}_r(\rho_{ABC})$ and $C^{A|B}_r(\rho_{AB})+C^{A|C}_r(\rho_{AC})$ is still unknown, that is, whether the following relation
holds for all tripartite states   remains to be verified,
\begin{eqnarray}\label{eq:mono_r}
C^{A|BC}_r(\rho_{ABC})\geq C^{A|B}_r(\rho_{AB})+C^{A|C}_r(\rho_{AC}).
\end{eqnarray}
We give an upper bound for the quantity $C^{A|BC}_r(\rho_{ABC})- C^{A|B}_r(\rho_{AB})-C^{A|C}_r(\rho_{AC})$ in terms of conditional entropy and find that the relation  \eqref{eq:mono_r} may not hold in general.

\begin{lem}\label{eq:Cr_vs_ne}
Given a tripartite state $\rho_{ABC}\in\cD(\cH_A\ot\cH_B\ot\cH_C)$, then
\begin{eqnarray}
\nonumber &&C^{A|BC}_r(\rho_{ABC})- C^{A|B}_r(\rho_{AB})-C^{A|C}_r(\rho_{AC})\\
&\leq& -S(A|BC)+S(A|B)+S(A|C),
\end{eqnarray}
where the conditional entropy is defined as $S(X|Y)=S(\rho_{XY})-S(\rho_Y)$.

\end{lem}
\begin{proof}
Let us take another system $\cH_{A'}=\cH_A$ and the local basis $\set{\ket{i}_{A'}}_i=\set{\ket{i}_A}_i$.
Define an unitary operator on $\cH_A\ot\cH_{A'}$ such that
$U_{AA'}\ket{i}_A\ket{0}_{A'}=\ket{i}_A\ket{i}_{A'}$.
Hence, for any tripartite state $\rho_{ABC}=\sum_{i,j}\Prok{i}{j}_A\ot\rho^{BC}_{ij}$,
\begin{eqnarray*}
\sigma_{ABCA'}&=&\Delta_A\ot\Delta_{A'}(U_{AA'}\rho_{ABC}\ot\proj{0}_{A'}U^{\dag}_{AA'})\\
&=&\sum_{i}\proj{i}_A\ot\rho^{BC}_{ii}\ot \proj{i}_{A'}.
\end{eqnarray*}
Thus, $\sigma_{AB}=\Ptr{A'C}{\sigma_{ABCA'}}=\sum_{i}\proj{i}_A\ot\rho^{B}_i$,
$\sigma_{A'C}=\Ptr{AB}{\sigma_{ABCA'}}=\sum_{i}\proj{i}_{A'}\ot\rho^{C}_i$, and
$S(\sigma_{AB})=S(\Delta_A(\rho_{AB}))$, $S(\sigma_{A'C})=S(\Delta_A(\rho_{AC}))$ where
$\rho^B_i=\Ptr{C}{\rho^{BC}_{ii}}$ and $\rho^C_{i}=\Ptr{B}{\rho^{BC}_{ii}}$. Then, as
$S(\sigma_{ABCA'})=S(\Delta_A(\rho_{ABC}))$, we have
\begin{eqnarray*}
&&S(\sigma_{ABCA'}||\sigma_{AB}\ot\sigma_{A'C})\\
&=&S(\rho_{AB})+S(\rho_{A'C})-S(\sigma_{ABCA'})\\
&=&S(\Delta_A(\rho_{AB}))+S(\Delta_A(\rho_{AC}))-S(\Delta_{A}(\rho_{ABC})).
\end{eqnarray*}
Since relative entropy is
monotone under partial trace,
then
\begin{eqnarray*}
S(\sigma_{ABCA'}||\sigma_{AB}\ot\sigma_{A'C})
\geq S(\sigma_{BC}||\sigma_B\ot\sigma_C),
\end{eqnarray*}
where $\sigma_{BC}=\rho_{BC}, \sigma_B=\rho_B$ and $\sigma_C=\rho_C$,
that is,
\begin{eqnarray*}
&&S(\Delta_A(\rho_{AB}))+S(\Delta_A(\rho_{AC}))-S(\Delta_{A}(\rho_{ABC}))\\
&\geq& S(\rho_B)+S(\rho_C)-S(\rho_{BC}).
\end{eqnarray*}
Therefore,
\begin{eqnarray*}
&&C^{A|BC}_r(\rho_{ABC})- C^{A|B}_r(\rho_{AB})-C^{A|C}_r(\rho_{AC})\\
&=&S(\rho_{AB})+S(\rho_{AC})-S(\rho_{ABC})\\
&&-[S(\Delta_A(\rho_{AB}))+S(\Delta_A(\rho_{AC}))-S(\Delta_{A}(\rho_{ABC}))]\\
&\leq&S(\rho_{AB})+S(\rho_{AC})-S(\rho_{ABC})\\
&&-[ S(\rho_B)+S(\rho_C)-S(\rho_{BC})]\\
&=&-S(A|BC)+S(A|B)+S(A|C).
\end{eqnarray*}

\end{proof}

The negative conditional entropy quantifies the amount of entanglement as
$S(A|B)<0$ indicates the entanglement between A and B \cite{Cerf1997}.
Thus, the following relation
\begin{eqnarray}
\label{eq:mono_con}-S(A|BC)\geq-S(A|B)-S(A|C),
\end{eqnarray}
can be viewed as a monogamy relation of entanglement, which holds for any pure tripartite state.
Besides, Lemma \ref{eq:Cr_vs_ne} illustrates that the violation of the relation
\eqref{eq:mono_con} will lead to the violation of the  relation
\eqref{eq:mono_r}.

\begin{prop}
There exists some tripartite state $\rho_{ABC}\in\cD(\cH_A\ot\cH_B\ot\cH_C)$ such that
\begin{eqnarray*}
C^{A|BC}_r(\rho_{ABC})\leq C^{A|B}_r(\rho_{AB})+C^{A|C}_r(\rho_{AC}).
\end{eqnarray*}
\end{prop}
\begin{proof}
It is easy to verify that the tripartite state with the form
$\rho_{ABC}=\rho_{A_1B}\ot\rho_{A_2B}$ violates the relation \eqref{eq:mono_con}
 where $\cH_A=\cH_{A_1}\ot\cH_{A_2}$. Thus, the relation \eqref{eq:mono_r}
does not hold in general.
\end{proof}

In view of the discussion in Ref. \cite{Bu2017c}, the relation \eqref{eq:mono_r} cannot hold in general as the term
$C^{A|B}_r(\rho_{AB})+C^{A|C}_r(\rho_{AC})$ contains two copies of local coherence $C_r(\rho_A)$, whereas the term $C^{A|BC}_r(\rho_{ABC})$ only contains
one copy of $C_r(\rho_A)$. The relation  \eqref{eq:mono_r} will be violated  for the tripartite state $\rho_{ABC}$
with weak correlation between B and C, e.g., $\rho_{ABC}=\rho_{A_1B}\ot\rho_{A_2B}$ where $\cH_A=\cH_{A_1}\ot\cH_{A_2}$.

By introducing smooth max and min-relative entropies of IQ coherence,
the distribution of coherence quantified by relative entropy in multipartite systems has been obtained in Ref. \cite{Bu2017c}.
Besides relative entropy of coherence, we  find the distribution of coherence of formation $C_f$ and assistance $C_a$ in bipartite
systems by introducing the corresponding IQ  coherence measures. The IQ coherence of formation on bipartite systems
is defined as follows,
\begin{eqnarray}
\nonumber C^{A|B}_f(\rho_{AB}):&=&\min\sum_ip_iC^{A|B}_r(\proj{\psi_i}_{AB})\\
&=&\min\sum_ip_iS(\Delta_A(\proj{\psi_i}_{AB})),
\end{eqnarray}
where the minimization is taken over all pure state decomposition ${\rho_{AB}=\sum_ip_i\proj{\psi_i}_{AB}}$. Since $S(\Delta_A(\proj{\psi_i}_{AB}))=S(\Delta_A(\Ptr{B}{\proj{\psi_i}_{AB}}))$  ( see Lemma \ref{lem:c_fd} in Appendix \ref{apen:cf} )
and von Neumann entropy is concave, then we have
\begin{eqnarray}\label{eq:ex_cf}
C^{A|B}_f(\rho_{AB})=\min
\sum_ip_iS[\Delta_A(\Ptr{B}{\rho^{AB}_{i}})],
\end{eqnarray}
where the minimization is taken over all state decomposition of $\rho_{AB}=\sum_ip_i\rho^{AB}_{i}$
 without the restriction of $\rho^{AB}_{i}$ to be pure state. $C^{A|B}_f$ satisfy the properties (i)-(v), where
 (i) and (v) are obvious, (iii) and (iv) are presented in Appendix \ref{apen:cf}, and (ii) comes directly from (iii) and (v).

Here, we consider the distribution of coherence of formation in bipartite systems in terms of  the
IQ coherence of formation $C^{A|B}_f$, where $C^{A|B}_f$ contains not only the local coherence in  subsystem but also the entanglement of formation $E_f$ \cite{Bennett1996} between A and B, for which we have
the following relation.
\begin{lem}\label{prop:cf1}
Given a bipartite state $\rho_{AB}\in\cD(\cH_A\ot\cH_B)$, then
\begin{eqnarray}
C^{A|B}_f(\rho_{AB})\geq C_r(\rho_A)+E_f(\rho_{AB}),
\end{eqnarray}
where $\rho_A$ is the reduced state on subsystem A, and $E_f(\rho_{AB})=\min \sum_i p_iS(\Ptr{A}{\proj{\psi_i}_{AB}})$ with the minimization being taken over all
pure state decomposition of $\rho_{AB}=\sum_ip_i\proj{\psi_i}_{AB}$.
\end{lem}
\begin{proof}
For any pure state decomposition of $\rho_{AB}=\sum_ip_i\proj{\psi_i}_{AB}$ with
$\rho^A_i=\Ptr{B}{\proj{\psi_i}_{AB}}$ and $\rho^B_i=\Ptr{A}{\proj{\psi_i}_{AB}}$, we have
\begin{eqnarray*}
\sum_{i}p_iS(\Delta_A(\proj{\psi_i}_{AB}))
&=&\sum_ip_i[C_r(\rho^A_i)+S(\rho^B_i)]\\
&\geq& C_r(\rho_A) +E_f(\rho_{AB}),
\end{eqnarray*}
where the first line comes from \eqref{prop:pure_cr}
and the second line comes from the convexity of $C_r$ and definition of $E_f$. Thus, we
 get the result.

\end{proof}

\begin{lem}\label{prop:cf2}
Given a bipartite state $\rho_{AB}\in\cD(\cH_A\ot\cH_B)$, then
\begin{eqnarray}
C_f(\rho_{AB})\geq C^{A|B}_f(\rho_{AB})+C_f(\rho_B),
\end{eqnarray}
where $\rho_B$ is the reduced state of $\rho_{AB}$ on subsystem B.
\end{lem}
\begin{proof}
There exists an optimal pure state decomposition of $\rho_{AB}=\sum_ip_i\proj{\psi_i}_{AB}$ such that
$C_f(\rho_{AB})=\sum_ip_iS(\Delta_A\ot\Delta_B(\proj{\psi_i}_{AB}))$, where $\ket{\psi_i}_{AB}=\sum_j\sqrt{\lambda_{i,j}}\ket{j}_A\ket{u_{i,j}}_B$
with $\sum_j\lambda_{i,j}=1$ for any i
and $\set{\ket{j}_A}_j$ is the reference basis of subsystem A. Thus,
\begin{eqnarray*}
C_f(\rho_{AB})&=&\sum_ip_iS(\Delta_A\ot\Delta_B(\proj{\psi_i}_{AB}))\\
&=&\sum_ip_i[S(\Delta_A(\rho^A_i))+\sum_j\lambda_{i,j}S(\Delta_B(\proj{u_{i,j}}_B))]\\
&=&\sum_ip_iS(\Delta_A(\rho^A_i))+\sum_{i,j}p_i\lambda_{i,j}S(\Delta_B(\proj{u_{i,j}}_B))\\
&\geq& C^{A|B}_f(\rho_{AB})+C_f(\rho_B),
\end{eqnarray*}
where the inequality results from the definitions of $C^{A|B}_f$ and $C_f$,  and the fact that
$\rho_B=\sum_{i,j}p_i\lambda_{i,j}\proj{u_{i,j}}_B$.
\end{proof}

Combining the above two lemmas, we can obtain  the distribution
of coherence of formation in bipartite systems, where the total coherence of formation is lower bounded by the sum
of  local coherence of formation in subsystems A and B and  the entanglement of formation between the subsystems.
\begin{thm}\label{thm:cf_ef}
Given a bipartite state $\rho_{AB}\in\cD(\cH_A\ot\cH_B)$,
it holds that,
\begin{eqnarray}
\nonumber C_f(\rho_{AB})\geq &&\max\set{C_r(\rho_A)+C_f(\rho_B), C_r(\rho_B)+C_f(\rho_A) }\\
&&~~~~~~~~~~~~+E_f(\rho_{AB}),
\end{eqnarray}
where $\rho_A$ and $\rho_B$ are the corresponding reduced states of $\rho_{AB}$.
\end{thm}
\begin{proof}
Based on Lemmas \ref{prop:cf1} and \ref{prop:cf2}, we have
\begin{eqnarray*}
C_f(\rho_{AB})&\geq& C^{A|B}_f(\rho_{AB})+C_f(\rho_B)\\
&\geq& C_r(\rho_A)+E_f(\rho_{AB})+C_f(\rho_B).
\end{eqnarray*}
Similarly, $C_f(\rho_{AB})\geq C_r(\rho_B)+C_f(\rho_A) +E_f(\rho_{AB})$ can be obtained.
\end{proof}

Now, we give an example such that the equality  in  Theorem \ref{thm:cf_ef}  holds.
For any quantum state $\rho_{B}\in\cD(\cH_B)$, there exists an optimal pure state decomposition
of $\rho_B=\sum_i p_i\proj{u_i}_B$ such that $C_f(\rho_B)=\sum_ip_iS(\Delta_B(\proj{u_i}))$. Let us take the pure bipartite state
$\ket{\psi}_{AB}=\sum_i\sqrt{p_i}\ket{i}_A\ket{u_i}_B$, then
$C_f(\psi_{AB})=C_r(\rho_A)+C_f(\rho_B)+E_f(\psi_{AB})$.

Besides, due to the equivalence between coherence of formation $C_f$ and coherence cost $C_c$ \cite{Winter2016}, we can obtain the
relationship between coherence cost $C_c$ and entanglement cost $E_c$ in bipartite systems from  Theorem \ref{thm:cf_ef}.
\begin{cor}
Given a bipartite state $\rho_{AB}\in\cD(\cH_A\ot\cH_B)$,
it holds that,
\begin{eqnarray}
\nonumber C_c(\rho_{AB})\geq &&\max\set{C_r(\rho_A)+C_c(\rho_B), C_r(\rho_B)+C_c(\rho_A) }\\
&&~~~~~~~~~~~~+E_c(\rho_{AB}),
\end{eqnarray}
where  $\rho_A$ and $\rho_B$ are the corresponding reduced states of $\rho_{AB}$, and the entanglement cost $E_c$ \cite{Hayden2001} is defined as
\begin{eqnarray*}
E_c(\rho)=\inf\set{t:\lim_{n\to\infty}\norm{\rho^{\ot n}-\Lambda_{LOCC}(\phi^{\ot tn}_+)}_{\mathrm{tr}}=0},
\end{eqnarray*}
with $\ket{\phi_+}=\frac{1}{\sqrt{2}}(\ket{00}-\ket{11})$,  $\Lambda_{LOCC}$ being the local operation and classical communication (LOCC) and trace norm $\norm{A}_{\mathrm{tr}}=\Tr{\sqrt{A^\dag A}}$.
\end{cor}
\begin{proof}
In view of Theorem \ref{thm:cf_ef}, we have the following relationship for
the bipartite state $\rho^{\ot n}_{AB}$,
\begin{eqnarray*}
 C_f(\rho^{\ot n}_{AB})\geq C_r(\rho^{\ot n}_A)+C_f(\rho ^{\ot n}_B)+E_f(\rho^{\ot n}_{AB}).
\end{eqnarray*}
Since both $C_r$ and $C_f$ are additive \cite{Winter2016} and  $E_c$ is equivalent to the regularized entanglement of formation $E_f$ \cite{Hayden2001},
we have
\begin{eqnarray*}
 C_f(\rho_{AB})\geq C_r(\rho_A)+C_f(\rho_B)+E_c(\rho_{AB}).
\end{eqnarray*}
Similarly, we can also obtain the following relation,
\begin{eqnarray*}
 C_f(\rho_{AB})\geq C_r(\rho_B)+C_f(\rho_A)+E_c(\rho_{AB}).
\end{eqnarray*}
Therefore, we obtain the result.

\end{proof}

It has been proved that relative entropy of coherence $C_r$ is equivalent to distillable coherence $C_d$ \cite{Winter2016}. Thus we can obtain the relationship between the
distillable coherence and distillable entanglement  in bipartite systems as follows.
\begin{cor}
Given a bipartite state $\rho_{AB}\in\cD(\cH_A\ot\cH_B)$,
$C_d(\rho_{AB})$ and  $E_d(\rho_{AB})$ has the following
relationship,
\begin{eqnarray}
C_d(\rho_{AB})\geq C_d(\rho_A)+C_d(\rho_B)+ E_d(\rho_{AB}),
\end{eqnarray}
where$\rho_A$ and $\rho_B$ are the corresponding reduced states of $\rho_{AB}$ and the  distillable entanglement $E_d$ \cite{Plenio2006} is defined as
\begin{eqnarray*}
E_d(\rho)=\inf\set{t :\lim_{n\to\infty}\norm{\Lambda_{LOCC}(\rho^{\ot n})-\phi^{\ot t n}_+}_{\mathrm{tr}}=0}.
\end{eqnarray*}
\end{cor}
\begin{proof}
It has been proved in Ref. \cite{Bu2017c} that
\begin{eqnarray*}
C_r(\rho_{AB})\geq C_r(\rho_A)+C_r(\rho_B)+E^{\infty}_r(\rho_{AB}),
\end{eqnarray*}
where $E^{\infty}_r$ is the regularized relative entropy of entanglement \cite{Vedral1997,Vedral1998,Vollbrecht2001}.
Due to the equivalence between $C_r$ and $C_d$ \cite{Winter2016}
and the fact that $E^{\infty}_r\geq E_d$ \cite{HorodeckiRMP09}, we obtain the result.

\end{proof}

In tripartite systems, the monogamy relation of coherence has been considered for relative entropy of coherence $C_r$ and it has been shown in Refs. \cite{Yao2015,Kumar2017} that
 it  does not hold in general  for $C_r$. However,
the monogamy relation for IQ coherence measure $C^{A|B}_r$ has been established in Ref. \cite{Bu2017c}. Here,
we obtain the monogamy relation for $C^{A|B}_f$ in tripartite systems as follows.

\begin{prop}
Given a bipartite state $\rho_{ABC}\in\cD(\cH_A\ot\cH_B\ot\cH_C)$, then
\begin{eqnarray}
C^{AB|C}_f(\rho_{ABC})
\geq C^{A|BC}_f(\rho_{ABC})+C^{B|C}_f(\rho_{BC}),
\end{eqnarray}
which implies the following monogamy relation,
\begin{eqnarray}
C^{AB|C}_f(\rho_{ABC})
\geq C^{A|C}_f(\rho_{AC})+C^{B|C}_f(\rho_{BC}).
\end{eqnarray}
\end{prop}
\begin{proof}
For any pure tripartite state $\ket{\psi}_{ABC}=\sum_i\sqrt{p_i}\ket{i}_A\ot\ket{u_i}_{BC}$,
\begin{eqnarray*}
&&C^{AB|C}_f(\psi_{ABC})\\
&=&C^{AB|C}_r(\psi_{ABC})\\
&=&S(\Delta_A\ot\Delta_B(\rho_{AB}))\\
&=&S(\Delta_A(\rho_A))+\sum_ip_iS(\Delta_{B}(\rho^B_i))\\
&=&S(\Delta_A(\rho_A))+\sum_ip_iS(\Delta_{B}(\mathrm{Tr}_C\proj{u_i}_{BC}))\\
&\geq& C^{A|BC}_f(\psi_{ABC})+C^{B|C}_f(\rho_{BC}),
\end{eqnarray*}
where the last inequality results from the fact that $S(\Delta_A(\rho_A))=C^{A|BC}_f(\psi_{ABC})$ for pure state
$\psi_{ABC}$ and $\sum_ip_iS(\Delta_{B}(\mathrm{Tr}_C\proj{u_i}_{BC}))\leq C^{B|C}_f(\rho_{BC})$ due to the definition of
$C^{B|C}_f$. For any tripartite states $\rho_{ABC}$, there exists an optimal
pure state decomposition of $\rho_{ABC}=\sum_{i}\lambda_i\proj{\psi_i}_{ABC}$ such that
$C^{AB|C}_f(\rho_{ABC})=\sum_i\lambda_iC^{AB|C}_f(\proj{\psi_i}_{ABC})$. Thus,
\begin{eqnarray*}
C^{AB|C}_f(\rho_{ABC})&=&\sum_i\lambda_iC^{AB|C}_f(\proj{\psi_i}_{ABC})\\
&\geq& \sum_i\lambda_i[C^{A|BC}_f(\proj{\psi_i}_{ABC})+C^{B|C}_f(\rho^{BC}_i)]\\
&\geq& C^{A|BC}_f(\rho_{ABC})+C^{B|C}_f(\rho_{BC}).
\end{eqnarray*}

\end{proof}

Similar to coherence of formation $C_f$, coherence of assistance $C_a$ is also defined by taking the pure state decompositions of the given state \cite{Chitambar2016}. Here, we introduce the IQ  coherence of assistance
$C^{A|B}_a$ on bipartite systems, which is defined as follows
\begin{eqnarray}
\nonumber C^{A|B}_a(\rho_{AB}):&=&\max\sum_ip_iC^{A|B}_r(\proj{\psi_i}_{AB})\\
&=&\max\sum_ip_iS(\Delta_A(\proj{\psi_i}_{AB})),
\end{eqnarray}
where the maximization is taken over all pure state decomposition ${\rho_{AB}=\sum_ip_i\proj{\psi_i}_{AB}}$.
Following the similar method, we can obtain the relationship between coherence of assistance $C_a$ and entanglement of assistance
$E_a$ \cite{DiVincenzo1999} in bipartite systems as follows.
\begin{thm}\label{thm:bi_Ca}
Given a bipartite state $\rho_{AB}\in\cD(\cH_A\ot\cH_B)$, it holds that,
\begin{eqnarray*}
C^{A|B}_a(\rho_{AB})&\leq& C_a(\rho_A)+E_a(\rho_{AB}),\\
C_a(\rho_{AB})&\leq& C^{A|B}_a(\rho_{AB})+C_a(\rho_B),\\
C_a(\rho_{AB})&\leq& C_a(\rho_A)+C_a(\rho_B)+E_a(\rho_{AB}),
\end{eqnarray*}
where $E_a(\rho_{AB})=\max \sum_i p_iS(\Ptr{A}{\proj{\psi_i}_{AB}})$ with the maximization being taken over all
pure state decomposition of $\rho_{AB}=\sum_ip_i\proj{\psi_i}_{AB}$ and $\rho_A, \rho_B$ are the reduced states of $\rho_{AB}$ on subsystems A and B, respectively.
\end{thm}
Theorem \ref{thm:bi_Ca}  illustrates that the total coherence of assistance
in bipartite systems is upper bounded by the sum of coherence of assistance
in each subsystem and the entanglement of formation between subsystems.
The proof of Theorem \ref{thm:bi_Ca} is almost the same as that of Theorem \ref{thm:cf_ef}, thus we omit it here.
The regularized version of coherence of assistance $C^\infty_a$ has also been proposed in
Ref. \cite{Chitambar2016}, which is defined as $C^\infty_a(\rho):=\lim_{n\to\infty}\frac{1}{n}C_a(\rho^{\ot n})=S(\Delta(\rho))$.
 Moreover, for any state
extension $\rho_{AB}$ of a given state $\rho_A$, i.e., $\Ptr{A}{\rho_{AB}}=\rho_A$,
$C^\infty_a(\rho_A)$ is a upper bound of $C^{A|B}_r(\rho_{AB})$. In fact,
$C^\infty_a(\rho_A)$ is the maximum value of $C^{A|B}_r(\rho_{AB})$ for the state extension $\rho_{AB}$
of $\rho_A$.
\begin{prop}
Given a quantum state $\rho_A\in\cD(\cH_A)$, then
\begin{eqnarray}
C^\infty_a(\rho_A)=\max_{\rho_{AB},\Ptr{B}{\rho_{AB}}=\rho_A}C^{A|B}_r(\rho_{AB}).
\end{eqnarray}
\end{prop}

\begin{proof}
First, for pure bipartite state $\psi_{AB}$ with $\Ptr{B}{\proj{\psi}_{AB}}=\rho_A$, then
$C^{A|B}_r(\psi_{AB})=S(\Delta_A(\rho_A))=C^\infty_a(\rho_A)$. Besides, for mixed bipartite state $\rho_{AB}$ with
$\Ptr{B}{\rho_{AB}}=\rho_A$, there exists a purification $\psi_{ABC}$ of $\rho_{AB}$ such that
$\Ptr{C}{\proj{\psi}_{ABC}}=\rho_{AB}$. Since $C^{A|B}_{r}$ is monotone under completely positive and trace preserving
(CPTP) maps on B side, then
$C^{A|B}_r(\rho_{AB})\leq C^{A|BC}_r(\psi_{ABC})=S(\Delta_A(\rho_A))=C^\infty_a(\rho_A)$.

\end{proof}

\section{$l_1$ norm of IQ coherence}
In order to introduce  $l_1$ norm of IQ coherence on bipartite systems, let us first introduce a new norm $\norm{\cdot}_{l_1\ot \mathrm{tr}}$ on $\cB(\cH_A\ot\cH_B)$ with a fixed basis $\set{\ket{i}_A}_i$ of $\cH_A$. For any operator $Q\in\cB(\cH_A\ot\cH_B)$ written as $Q=\sum^{d_A}_{i,j=1}\Prok{i}{j}_A\ot Q^B_{ij}$,
the norm $\norm{Q}_{l_1\ot \mathrm{tr}}$ is defined as follows,
\begin{eqnarray}
\norm{Q}_{l_1\ot \mathrm{tr}}:=\sum^{d_A}_{i,j=1}\norm{Q^B_{ij}}_{\mathrm{tr}},
\end{eqnarray}
where $\norm{A}_{\mathrm{tr}}=\Tr{\sqrt{A^\dag A}}$.
It is easy to show that $\norm{\cdot}_{l_1\ot \mathrm{tr}}$ is a norm, that is, it satisfies the following properties:
(i) Positivity, $\norm{Q}_{l_1\ot \mathrm{tr}}\geq 0$ and $\norm{Q}_{l_1\ot \mathrm{tr}}=0\Leftrightarrow Q=0$;
(ii) $\norm{\alpha Q}_{l_1\ot \mathrm{tr}}=|\alpha|\norm{Q}_{l_1\ot \mathrm{tr}}$ for any $\alpha\in\complex$;
(iii) Triangle inequality, $\norm{Q+P}_{l_1\ot \mathrm{tr}}\leq \norm{Q}_{l_1\ot \mathrm{tr}}+\norm{P}_{l_1\ot \mathrm{tr}}$ for any operators $Q,P\in\cB(\cH_A\ot\cH_B)$.

Based on this new norm, we define  $l_1$ norm of IQ coherence on bipartite systems as follows,
\begin{eqnarray}
\nonumber C^{A|B}_{l_1}(\rho_{AB})
&:=&\min_{\sigma_{A|B}\in\cI Q}\norm{\rho_{AB}-\sigma_{A|B}}_{l_1\ot \mathrm{tr}}\\
&=&\sum_{i\neq j}\norm{\rho^B_{ij}}_{\mathrm{tr}},
\end{eqnarray}
where $\rho^B_{ij}=\Innerm{i}{\rho_{AB}}{j}_A$.
Note that, $C^{A|B}_{l_1}$ satisfies the
properties (i)-(v), where
 the positivity of $C^{A|B}_{l_1}$ ( i.e., property (i) ) comes from the positivity
of the norm $\norm{\cdot}_{l_1\ot \mathrm{tr}}$, (iv) results from the contractivity  of $\norm{\cdot}_{\mathrm{tr}}$ under
CPTP maps, (v) comes from the
triangle inequality of the norm $\norm{\cdot}_{l_1\ot \mathrm{tr}}$, (iii) and (v) lead to the property (ii). Thus, we only need to prove (iii), which is presented
in the Appendix \ref{lem:str_mon}.

Due to the definition,  $C^{A|B}_{l_1}(\rho_{AB})\geq C_{l_1}(\rho_A)$ with $\rho_A$ being the reduced state of $\rho_{AB}$, which comes
from the fact that $\norm{\rho^B_{ij}}_{\mathrm{tr}}\geq \abs{\Tr{\rho^{B}_{ij}}}$.
If the subsystem $B$ is a trivial system, i.e., $dim \cH_B=1$, then $C^{A|B}_{l_1}(\rho_{AB})$ reduces to $C_{l_1}(\rho_A)$.
Besides, for bipartite pure state $\ket{\psi}_{AB}$, which can be written as $\ket{\psi}_{AB}=\sum^{d_A}_{i=1}\sqrt{p_i}\ket{i}_A\ket{u_i}_B$,
$C^{A|B}_{l_1}(\psi_{AB})=(\sum^{d_A}_{i=1}\sqrt{p_i})^2-1\leq d_A-1$.
Thus, the maximum value for $C^{A|B}_{l_1}$ is $d_A-1$ which does not depend on the subsystem B.

\begin{prop}
Given an bipartite state $\rho_{AB}\in\cD(\cH_A\ot\cH_B)$, then
\begin{eqnarray}\label{eq:sub_l1}
C_{l_1}(\rho_{AB})\geq C^{A|B}_{l_1}(\rho_{AB})+C_{l_1}(\rho_B),
\end{eqnarray}
where $\rho_B$ is the reduced state of $\rho_{AB}$.
\end{prop}
\begin{proof}
For any bipartite state $\rho_{AB}=\sum_{i,j}\Prok{i}{j}_A\ot \rho^{B}_{ij}$,
 the reduced state $\rho_B$ can be written as $\rho_B=\sum_i\rho^{B}_{ii}$. Thus
\begin{eqnarray*}
C^{A|B}_{l_1}(\rho_{AB})&=&\sum_{i\neq j}\norm{\rho^B_{ij}}_{\mathrm{tr}},\\
C_{l_1}(\rho_B)&=&\norm{\sum_{i}\rho^B_{ii}}_{l_1}-1=\sum_{j\neq k}|\sum_i\Innerm{j}{\rho^B_{ii}}{k}_B|,\\
C_{l_1}(\rho_{AB})&=&\sum_{j\neq k}\sum_i|\Innerm{j}{\rho^B_{ii}}{k}_B|+\sum_{i\neq j}\norm{\rho^B_{ij}}_{l_1}.
\end{eqnarray*}
Since $\sum_{j\neq k}|\sum_i\Innerm{j}{\rho^B_{ii}}{k}_B|\leq \sum_{j\neq k}\sum_i|\Innerm{j}{\rho^B_{ii}}{k}_B|$
and $\sum_{i\neq j}\norm{\rho^B_{ij}}_{\mathrm{tr}}\leq \sum_{i\neq j}\norm{\rho^B_{ij}}_{l_1}$ ( See Lemma \ref{lem:11_vs_tr} in Appendix \ref{appen:l1_tr}), we get the result.

\end{proof}

This relation \eqref{eq:sub_l1} is stronger than the known result $C_{l_1}(\rho_{AB})\geq C_{l_1}(\rho_A)+C_{l_1}(\rho_B)$ as
$C^{A|B}_{l_1}(\rho_{AB})\geq C_{l_1}(\rho_A)$. Besides, $C^{A|B}_{l_1}$  contains not only the local coherence in subsystem A but also
the nonlocal correlation between A and B from the following proposition.

\begin{prop}
Given a bipartite state $\rho_{AB}\in\cD(\cH_A\ot\cH_B)$, then we have the following relationship,
\begin{eqnarray}
C^{A|B}_{l_1}(\rho_{AB})^2-C_{l_1}(\rho_A)^2\geq 2(\Tr{\rho^2_{AB}}-\Tr{\rho^2_A}).
\end{eqnarray}

\end{prop}

\begin{proof}
Any bipartite state $\rho_{AB}$
can be written as $\rho_{AB}=\sum^{d_A}_{i,j=1}\Prok{i}{j}_A\ot \rho^B_{ij}$ and thus
\begin{eqnarray*}
C^{A|B}_{l_1}(\rho_{AB})=2\sum_{i<j}\Tr{|\rho^B_{ij}|},\\
C_{l_1}(\rho_A)=2\sum_{i<j}|\Tr{\rho^B_{ij}}|.
\end{eqnarray*}

Moreover, the term $\Tr{\rho^2_{AB}}-\Tr{\rho^2_A}$
has the following upper bound,
\begin{eqnarray*}
&&\Tr{\rho^2_{AB}}-\Tr{\rho^2_A}\\
&=&(\sum_i\Tr{|\rho^B_{ii}|^2}+2\sum_{i<j}\Tr{|\rho^B_{ij}|^2})\\
&&-(\sum_i\Tr{\rho^B_{ii}}^2+2\sum_{i<j}|\Tr{\rho^B_{ij}}|^2)\\
&\leq& 2\sum_{i<j}(\Tr{|\rho^B_{ij}|^2}-|\Tr{\rho^B_{ij}}|^2)\\
&\leq& 2\sum_{i<j}(\Tr{|\rho^B_{ij}|}^2-|\Tr{\rho^B_{ij}}|^2),
\end{eqnarray*}
where the first and the second inequalities come from the fact that $\Tr{|\rho^B_{ij}|^2}\leq \Tr{|\rho^B_{ij}|}^2$.
Therefore,
\begin{eqnarray*}
&&C^{A|B}_{l_1}(\rho_{AB})^2-C_{l_1}(\rho_A)^2\\
&=&4[(\sum_{i<j}\Tr{|\rho^B_{ij}|})^2-(\sum_{i<j}|\Tr{\rho^B_{ij}}|)^2]\\
&=&4[\sum_{i<j}(\Tr{|\rho^B_{ij}|}-|\Tr{\rho^B_{ij}}|)][\sum_{i<j}(\Tr{|\rho^B_{ij}|}+|\Tr{\rho^B_{ij}}|)]\\
&\geq& 4 \sum_{i<j} [\Tr{|\rho^B_{ij}|}^2-|\Tr{\rho^B_{ij}}|^2]\\
&\geq& 2(\Tr{\rho^2_{AB}}-\Tr{\rho^2_A}),
\end{eqnarray*}
where the first inequality comes directly from the fact that $\Tr{|\rho^B_{ij}|}\leq \Tr{|\rho^B_{ij}|}$  and the
second inequality comes from the upper bound of $\Tr{\rho^2_{AB}}-\Tr{\rho^2_A}$.

\end{proof}

The term $\Tr{\rho^2_{AB}}-\Tr{\rho^2_A}$ quantifies the entanglement between A and B as
\begin{eqnarray}\label{ineq:ent}
\Tr{\rho^2_{AB}}-\Tr{\rho^2_A}>0
\end{eqnarray}
only if $\rho_{AB}$ is entangled \cite{Horodecki1996} and the inequality \eqref{ineq:ent}
provides a powerful tool in the detection of entanglement in  experiments \cite{Bovino2005,Islam2015}.
Thus, the above proposition implies that the total coherence in bipartite system  quantified by $l_1$ norm consists of the  nonlocal correlation between A and B and
the local coherence $C_{l_1}(\rho_A)$ and $C_{l_1}(\rho_B)$.
Furthermore,  we obtain the monogamy relation of $C^{A|B}_{l_1}$ in tripartite systems,
which clarifies the distribution of coherence by $l_1$ norm in multipartite systems.
\begin{prop}
Given a tripartite state $\rho_{ABC}\in\cD(\cH_A\ot\cH_B\ot\cH_C)$, then
\begin{eqnarray}
C^{AB|C}_{l_1}(\rho_{ABC})\geq C^{A|BC}_{l_1}(\rho_{ABC})+C^{B|C}_{l_1}(\rho_{BC}),
\end{eqnarray}
which implies the following monogamy relation,
\begin{eqnarray}
C^{AB|C}_{l_1}(\rho_{ABC})\geq C^{A|C}_{l_1}(\rho_{AC})+C^{B|C}_{l_1}(\rho_{BC}),
\end{eqnarray}
where $\rho_{AC}, \rho_{BC}$ are the corresponding reduced states of $\rho_{ABC}$.

\end{prop}
\begin{proof}
Any tripartite state $\rho_{ABC}$ can be written  as $\rho_{ABC}=\sum_{i,j}\sum_{m,n}\Prok{i}{j}_A\ot \Prok{m}{n}_B\ot\rho^C_{ij,mn}$
with the local basis $\set{\ket{i}_A}_i$ and $\set{\ket{m}_B}_m$ of $\cH_A$ and $\cH_B$. Then the reduced state $\rho_{BC}=\sum_i\sum_{m,n}\Prok{m}{n}_B\ot \rho^C_{ii,mn}$. Thus
\begin{eqnarray*}
C^{AB|C}_{l_1}(\rho_{ABC})&=&\sum_{(i,m)\neq (j,n)}\norm{\rho^C_{ij,mn}}_{\mathrm{tr}},\\
C^{A|BC}_{l_1}(\rho_{ABC})&=&\sum_{i\neq j}\norm{\sum_{m,n}\Prok{m}{n}_B\ot\rho^C_{ij,mn}}_{\mathrm{tr}}\\
&\leq& \sum_{i\neq j}\sum_{m,n}\norm{\Prok{m}{n}_B\ot\rho^C_{ij,mn}}_{\mathrm{tr}}\\
&=&\sum_{i\neq j}\sum_{m,n}\norm{\rho^C_{ij,mn}}_{\mathrm{tr}},\\
C^{B|C}_{l_1}(\rho_{BC})&=&\sum_{m\neq n}\norm{\sum_{i}\rho^C_{ii,mn}}_{\mathrm{tr}}
\leq \sum_i\sum_{m\neq n}\norm{\rho^C_{ii,mn}}_{\mathrm{tr}},
\end{eqnarray*}
where $(i,m)\neq (j,n)$ means $i\neq j$ or $m\neq n$. Therefore, we get the result.

\end{proof}

Now, let us consider the relationship between $C^{A|B}_{l_1}$ and $C^{A|B}_{\max}$, where
we find that $C^{A|B}_{l_1}$ is closely related to $C^{A|B}_{\max}$ and
they are equal for certain type of bipartite states.
\begin{prop}\label{prop:l1vsMax}
Given a bipartite state $\rho_{AB}\in\cD(\cH_A\ot\cH_B)$, then
\begin{eqnarray}
\nonumber 1+\frac{1}{d_A-1}C^{A|B}_{l_1}(\rho_{AB})\leq 2^{C^{A|B}_{\max}(\rho_{AB})}\leq 1+C^{A|B}_{l_1}(\rho_{AB}),\\
\end{eqnarray}
where $d_A$ is the dimension of $\cH_A$.
\end{prop}

\begin{proof}
Any bipartite state $\rho_{AB}$ can be written as $\rho_{AB}=\sum_{i,j}\Prok{i}{j}_A\ot\rho^B_{ij}$, then
$C^{A|B}_{l_1}(\rho_{AB})=\sum_{i\neq j}\norm{\rho^B_{ij}}_{\mathrm{tr}}=\sum_{i,j}\norm{\rho^B_{ij}}_{\mathrm{tr}}-1$.
For $\rho^B_{ij}$ with $i<j$, there exists a unitary $U^B_{ij}$ such that
$U^B_{ij}\rho^B_{ij}=|\rho^B_{ij}|$, and thus
$\rho^B_{ji}U^{B\dag}_{ij}=|\rho^B_{ij}|$ as $(\rho^B_{ij})^\dag=\rho^B_{ji}$.
Now, let us take the  positive operator $M$ as follows,
\begin{eqnarray*}
M&=&\frac{1}{d_A-1}\sum_{i<j}[\proj{i}_A\ot \mathbb{I}_B+\proj{j}_A\ot \mathbb{I}_B\\
&&~~~~~~~+\Prok{i}{j}_A\ot U^{B\dag}_{ij}+\Prok{j}{i}_A\ot U^B_{ij}]\\
&=&\mathbb{I}_{AB}+\frac{1}{d_A-1}[\Prok{i}{j}_A\ot U^{B\dag}_{ij}+\Prok{j}{i}_A\ot U^B_{ij}],
\end{eqnarray*}
where the positivity of $M$ comes from the fact that
$\proj{i}_A\ot \mathbb{I}_B+\proj{j}_A\ot \mathbb{I}_B=|\Prok{i}{j}_A\ot U^{B,\dag}_{ij}+\Prok{j}{i}_A\ot U^B_{ij}|$ and the
fact that $|X|+X\geq 0$ for any Hermitian operator $X$.

Due to the definition of $C^{A|B}_{\max}$, there exists an incoherent-quantum state
$\tau_{A|B}$ such that
\begin{eqnarray*}
\rho_{AB}\leq 2^{C^{A|B}_{\max}(\rho_{AB})}\tau_{A|B}.
\end{eqnarray*}
Thus
\begin{eqnarray*}
\Tr{M\rho_{AB}}\leq 2^{C^{A|B}_{\max}(\rho_{AB})}\Tr{M\tau_{A|B}},
\end{eqnarray*}
which leads to
\begin{eqnarray*}
1+\frac{1}{d_A-1}C^{A|B}_{l_1}(\rho_{AB})\leq 2^{C^{A|B}_{\max}(\rho_{AB})}.
\end{eqnarray*}

Besides, let us take the incoherent-quantum state $\sigma_{A|B}$ to be
\begin{eqnarray*}
\sigma_{A|B}&=&\frac{1}{1+C^{A|B}_{l_1}(\rho_{AB})}[\sum_{i}\proj{i}_A\ot \rho^B_{ii}\\
&&+\sum_{i<j}\proj{i}_A\ot |\rho^B_{ji}|+\proj{j}_A\ot|\rho^B_{ij}|].\\
\end{eqnarray*}
Then $\rho_{AB}\leq (1+C^{A|B}_{l_1}(\rho_{AB}))\sigma_{A|B}$, as
\begin{eqnarray*}
&& (1+C^{A|B}_{l_1}(\rho_{AB}))\sigma_{A|B}-\rho_{AB}\\
 &=&\sum_{i<j}[\proj{i}_A\ot |\rho^B_{ji}|+\proj{j}_A\ot|\rho^B_{ij}|\\
&& -\Prok{i}{j}_A\ot\rho^B_{ij}-\Prok{j}{i}_A\ot\rho^B_{ji}]\\
&\geq&0,
\end{eqnarray*}
where the inequality comes from the fact that
$\proj{i}_A\ot |\rho^B_{ji}|+\proj{j}_A\ot|\rho^B_{ij}|=\abs{\Prok{i}{j}_A\ot\rho^B_{ij}+\Prok{j}{i}_A\ot\rho^B_{ji}}$ for $i\neq j$.

\end{proof}

\begin{prop}\label{prop:eq_l1_max}
Given a bipartite state $\rho_{AB}\in\cD(\cH_A\ot\cH_B)$, if there exists a unitary operator
$U_{AB}=\sum_{i}\proj{i}_A\ot U^B_{i}$ such that
$U^B_i\rho^{B}_{ij}U^{B,\dag}_j=|\rho^B_{ij}|$ for any i,j, then
\begin{eqnarray}\label{eq:l1_max}
C^{A|B}_{\max}(\rho_{AB})= \log(1+C^{A|B}_{l_1}(\rho_{AB})),
\end{eqnarray}
 where $|P|$ is defined as $|P|=\sqrt{P^\dag P}$.

\end{prop}
\begin{proof}
Set $\lambda=2^{C^{A|B}_{\max}(\rho_{AB})}$. Due to the definition of
$C^{A|B}_{\max}$, there exists a state $\sigma_{A|B}\in\cI Q$ such that
\begin{eqnarray*}
\rho_{AB}\leq 2^{\lambda}\sigma_{A|B}.
\end{eqnarray*}
Applying the unitary operation
$U_{AB}(\cdot) U^\dag_{AB}$ on both sides of the above equation and the taking the partial trace on part B, one obtains
\begin{eqnarray*}
\sum_{i,j}\norm{\rho^B_{ij}}_{\mathrm{tr}}\Prok{i}{j}_A\leq 2^{\lambda}\sigma_A,
\end{eqnarray*}
where $\sigma_A$ is the reduced state of $\sigma_{A|B}$ and thus $\sigma_A\in \cI$.
Taking the pure state $\ket{+}=\frac{1}{\sqrt{d_A}}\sum_i\ket{i}_A$, we get
\begin{eqnarray*}
\bra{+}\sum_{i,j}\norm{\rho^B_{ij}}_{\mathrm{tr}}\Prok{i}{j}_A\ket{+}\leq 2^{\lambda}\bra{+}\sigma_A\ket{+},
\end{eqnarray*}
which implies that $2^{\lambda}\geq \sum_{i,j}\norm{\rho^B_{ij}}_{\mathrm{tr}}$ as $\Innerm{+}{\sigma_\lambda}{+}={1}/{d_A}$ and
$\iinner{+}{i}\iinner{j}{+}={1}/{d_A}$ for any $i,j$. That is,
$C^{A|B}_{\max}(\rho_{AB})\geq \log(\sum_{i,j}\norm{\rho^B_{ij}}_{\mathrm{tr}})=\log(1+C^{A|B}_{l_1}(\rho_{AB}))$.
Combining with Proposition \ref{prop:l1vsMax}, we obtain the result.

\end{proof}

It is easy to see that pure bipartite states satisfy the conditions in Proposition \ref{prop:eq_l1_max}. Thus
the equation \eqref{eq:l1_max} holds for any bipartite pure states. Moreover, the bipartite states $\rho_{AB}$, which have the following
form  $\rho_{AB}=\sum_k\rho^{AB}_k$ with $\rho^{AB}_{k}=\proj{k}\ot \rho^{B}_{kk}+\proj{d_A-k}\ot \rho^{B}_{d_A-k,d_A-k}+\Prok{k}{d_A-k}\ot\rho^B_{k,d_A-k}
+\Prok{d_A-k}{k}\ot \rho^B_{d_A-k,k}$, also satisfy the conditions in Proposition \ref{prop:eq_l1_max}, that is, the equation \eqref{eq:l1_max} holds for
such states. For example, for the bipartite state $\rho_{AB}(\lambda)=\lambda\proj{\phi_+}+(1-\lambda)\proj{\psi_+}$ with $0\leq\lambda\leq1$,
$\ket{\phi_+}=\frac{1}{\sqrt{2}}(\ket{00}-\ket{11})$ and $\ket{\psi_+}=\frac{1}{\sqrt{2}}(\ket{01}-\ket{10})$, the  equation \eqref{eq:l1_max} holds.

Note that, other coherence measures defined on a single system, such as coherence weight \cite{Bu2017asym},  can also be
used to define the corresponding IQ coherence measures on bipartite systems in a similar way, which is omitted here.
Although the IQ  coherence measure depends on the local basis in subsystem A, it will become
the measures of classical-quantum correlation if we take the minimization over  all the local
basis on system A \cite{Nakano2013,AdessoJPA2016}. For example, let us take the minimization over all the local basis for  $l_1$ norm of IQ coherence as
follows,
\begin{eqnarray*}
Q^{A|B}_{l_1}(\rho_{AB})=\min_{\text{local basis on A}}C^{A|B}_{l_1}(\rho_{AB}),
\end{eqnarray*}
where $Q^{A|B}_{l_1}$ is called one-side negativity of quantumness \cite{Nakano2013,AdessoJPA2016}.

\section{Additivity of IQ coherence measures}
The above sections show that IQ coherence measures can capture the nonlocal correlation between subsystems. However, the measure of nonlocal
correlation may not be additive, such as the relative entropy of entanglement. Thus we discuss the additivity of IQ coherence measures
in this section. Let us begin with the simplest case, relative entropy and $l_1$ norm. In view of the definition, it is easy to see the additivity of
$C^{A|B}_r$ and $C^{A|B}_{l_1}$: for any two bipartite states $\rho_{A_1B_1}\in\cD(\cH_{A_1}\ot\cH_{B_1})$ and $\rho_{A_2B_2}\in\cD(\cH_{A_2}\ot\cH_{B_2})$,
then
\begin{eqnarray*}
\nonumber &&C^{A_1A_2|B_1B_2}_r(\rho_{A_1B_1}\ot\rho_{A_2B_2})\\
&=&C^{A_1|B_1}_r(\rho_{A_1B_1})+ C^{A_2|B_2}_f(\rho_{A_2B_2}),
\end{eqnarray*}
and
\begin{eqnarray*}
\nonumber &&1+C^{A_1A_2|B_1B_2}_{l_1}(\rho_{A_1B_1}\ot\rho_{A_2B_2})\\
&=&[1+C^{A_1|B_1}_{l_1}(\rho_{A_1B_1})]\cdot[1+C^{A_2|B_2}_{l_1}(\rho_{A_2B_2})].
\end{eqnarray*}

Now, we consider the additivity of IQ coherence measures $C^{A|B}_{\max}$ and $C^{A|B}_f$, for which we have the following propositions.
\begin{prop}
For any two bipartite states $\rho_{A_1B_1}\in\cD(\cH_{A_1}\ot\cH_{B_1})$ and $\rho_{A_2B_2}\in\cD(\cH_{A_2}\ot\cH_{B_2})$,
\begin{eqnarray}
\nonumber &&C^{A_1A_2|B_1B_2}_{\max}(\rho_{A_1B_1}\ot\rho_{A_2B_2})\\
&=&C^{A_1|B_1}_{\max}(\rho_{A_1B_1})+ C^{A_2|B_2}_{\max}(\rho_{A_2B_2}).
\end{eqnarray}

\end{prop}
\begin{proof}

Due to definition of max-relative entropy of IQ coherence measure, there exists optimal $\cI Q$ states $\sigma_{A_1|B_1}$ and $\sigma_{A_2|B_2}$
such that $\rho_{A_iB_i}\leq 2^{C^{A|B}_{\max}(\rho_{A_iB_i})}\sigma_{A_i|B_i}$.
Hence, we have the following inequality,
\begin{eqnarray*}
\nonumber &&C^{A_1A_2|B_1B_2}_{\max}(\rho_{A_1B_1}\ot\rho_{A_2B_2})\\
&\leq&C^{A_1|B_1}_{\max}(\rho_{A_1B_1})+ C^{A_2|B_2}_{\max}(\rho_{A_2B_2}).
\end{eqnarray*}

Now, we prove the converse.
It has been proved in Ref. \cite{Bu2017c} that
\begin{eqnarray*}
2^{C^{A|B}_{\max}(\rho_{AB})}=\max_{\substack{
\tau_{AB}\geq0\\
\Delta_{A}\ot \mathbb{I}_B(\tau_{AB})=\mathbb{I}_{AB}
}}\Tr{\rho_{AB}\tau_{AB}}.
\end{eqnarray*}
Hence, there exist operators $\tau_{A_iB_i}$ such that
$\tau_{A_iB_i}\geq0$, $\Delta_{A_i}\ot \mathbb{I}_{B_i}(\tau_{A_iB_i})=\mathbb{I}_{A_iB_i}$ and
$
2^{C^{A_i|B_i}_{\max}(\rho_{AB})}=\Tr{\rho_{A_iB_i}\tau_{A_iB_i}},
$
for $i=1,2$. Then the operator $\tau_{A_1A_2B_1B_2}:=\tau_{A_1B_1}\ot\tau_{A_2B_2}$ satisfies the conditions
$\tau_{A_1A_2B_1B_2}\geq 0$ and $\Delta_{A_1}\ot \Delta_{A_2}\ot \mathbb{I}_{B_1B_2}(\tau_{A_1A_2B_1B_2})=\mathbb{I}_{A_1A_2B_1B_2}$, which implies that
$2^{C^{A_1A_2|B_1B_2}_{\max}(\rho_{A_1B_1}\ot\rho_{A_2B_2})}\geq \Tr{(\rho_{A_1B_1}\ot \rho_{A_2B_2})\tau_{A_1A_2B_1B_2}}$, i.e.,
\begin{eqnarray*}
\nonumber &&C^{A_1A_2|B_1B_2}_{\max}(\rho_{A_1B_1}\ot\rho_{A_2B_2})\\
&\geq&C^{A_1|B_1}_{\max}(\rho_{A_1B_1})+ C^{A_2|B_2}_{\max}(\rho_{A_2B_2}).
\end{eqnarray*}
Therefore, we obtain the result.

\end{proof}

\begin{prop}
For any two bipartite states $\rho_{A_1B_1}\in\cD(\cH_{A_1}\ot\cH_{B_1})$ and $\rho_{A_2B_2}\in\cD(\cH_{A_2}\ot\cH_{B_2})$,
\begin{eqnarray}
\nonumber &&C^{A_1A_2|B_1B_2}_f(\rho_{A_1B_1}\ot\rho_{A_2B_2})\\
&=&C^{A_1|B_1}_f(\rho_{A_1B_1})+ C^{A_2|B_2}_f(\rho_{A_2B_2}).
\end{eqnarray}

\end{prop}

\begin{proof}
Due to the definition of $C^{A|B}_f$, it is easy to get the inequality
\begin{eqnarray*}
&&C^{A_1A_2|B_1B_2}_f(\rho_{A_1B_1}\ot\rho_{A_2B_2})\\
&\leq& C^{A_1|B_1}_f(\rho_{A_1B_1})+ C^{A_2|B_2}_f(\rho_{A_2B_2}).
\end{eqnarray*}
Thus, we only need to prove the converse. First, we prove that
for any pure state $\ket{\psi}_{A_1A_2B_1B_2}$, the following inequality holds,
\begin{eqnarray}
\nonumber &&C^{A_1A_2|B_1B_2}_f(\psi_{A_1A_2B_1B_2})\\
\label{ineq:c_f_ad}&\geq& C^{A_1|B_1}_f(\sigma_{A_1B_1})+ C^{A_2|B_2}_f(\sigma_{A_2B_2}),
\end{eqnarray}
where $\sigma_{A_1B_1}, \sigma_{A_2B_2}$ are the corresponding reduced states of $\proj{\psi}_{A_1A_2B_1B_2}$.
Since the pure state $\ket{\psi}_{A_1A_2B_1B_2}$ can be written as
$\ket{\psi}_{A_1A_2B_1B_2}=\sum_i\sqrt{p_i}\ket{i}_{A_1}\ket{u_i}_{A_2B_1B_2}$, then
\begin{eqnarray*}
&&C^{A_1A_2|B_1B_2}_f(\psi_{A_1A_2B_1B_2})\\
&=&S(\Delta_{A_1}\ot\Delta_{A_2}(\sigma_{A_1A_2}))\\
&=&S(\Delta_{A_1}(\sigma_{A_1}))+\sum_ip_iS(\Delta_{A_2}(\sigma^{A_2}_i))\\
&=&S(\Delta_{A_1}(\sigma_{A_1}))+\sum_ip_iS(\Delta_{A_2}(\mathrm{Tr}_{B_2}\sigma^{A_2B_2}_i))\\
&\geq&C^{A_1|B_1}_f(\sigma_{A_1B_1})+C^{A_2|B_2}_f(\sigma_{A_2B_2}),
\end{eqnarray*}
where
$S(\Delta_{A_1}(\sigma_{A_1}))\geq C^{A_1|B_1}_f(\sigma_{A_1B_1}) $
results from the concavity of von Neumann entropy, and
$\sum_ip_iS(\Delta_{A_2}(\mathrm{Tr}_{B_2}\sigma^{A_2B_2}_i))\geq C^{A_2|B_2}_f(\sigma_{A_2B_2})$
comes from \eqref{eq:ex_cf}.

Moreover, there exists an optimal pure state decomposition of $\rho_{A_1B_1}\ot\rho_{A_2B_2}=\sum_i\lambda_i\proj{\psi_i}_{A_1A_2B_1B_2}$ such that
$C^{A_1A_2|B_1B_2}_f(\rho_{A_1B_1}\ot\rho_{A_2B_2})=\sum_i\lambda_iC^{A_1A_2|B_1B_2}_f(\proj{\psi_i}_{A_1A_2B_1B_2})$. Therefore,
\begin{eqnarray*}
&&C^{A_1A_2|B_1B_2}_f(\rho_{A_1B_1}\ot\rho_{A_2B_2})\\
&=&\sum_i\lambda_iC^{A_1A_2|B_1B_2}_f(\proj{\psi_i}_{A_1A_2B_1B_2})\\
&\geq&\sum_ip_i[C^{A_1|B_1}_f(\sigma^{A_1B_1}_i)+C^{A_2|B_2}_f(\sigma^{A_2B_2}_i)]\\
&\geq& C^{A_1|B_1}_f(\rho_{A_1B_1})+C^{A_2|B_2}_f(\rho_{A_2B_2}),
\end{eqnarray*}
where the first inequality comes from \eqref{ineq:c_f_ad}, the second inequality comes from the convexity
of $C^{A|B}_f$ with $\sigma^i_{A_1B_1}=\mathrm{Tr}_{A_2B_2}\proj{\psi_i}_{A_1A_2B_1B_2}, \sigma^{A_2B_2}_i=\mathrm{Tr}_{A_1B_1}\proj{\psi_i}_{A_1A_2B_1B_2}$, $\rho_{A_1B_1}=\sum_ip_i\sigma^{A_1B_1}_i$ and $\rho_{A_2B_2}=\sum_ip_i\sigma^{A_2B_2}_i$.

\end{proof}

Note that the additivity of $C^{A|B}_a$ is still unclear as the method used in the proof of the additivity of $C^{A|B}_f$
does not work for $C^{A|B}_a$.
Nevertheless, if the subsystems $B_i$ ( $i=1,2$ ) are trivial, i.e., the dimension is $1$, then
one has the additivity of the coherence measures. For example,   the additivity of IQ coherence measures $C^{A|B}_{\max}$ will
lead to the additivity of $C_{\max}$ if the subsystems $B_i$ ( $i=1,2$ ) are trivial.
\begin{cor}
Given two quantum states $\rho_{1}\in\cD(\cH_{A_1})$ and
$\rho_2\in\cD(\cH_{A_2})$, it holds that
\begin{eqnarray}
C_{\max}(\rho_1\ot\rho_2)=C_{\max}(\rho_1)+C_{\max}(\rho_2).
\end{eqnarray}
\end{cor}
Due to the additivity of $C_{\max}$, we  can obtain the additivity of
robustness of coherence $ROC$ \cite{Napoli2016} as follows,
\begin{eqnarray}
\nonumber1+ROC(\rho_1\ot \rho_2)
=[1+ROC(\rho_1)]\cdot[1+ROC(\rho_2)],\\
\end{eqnarray}
which comes directly from the fact that
$C_{\max}(\rho)=\log (1+ROC(\rho))$ \cite{Bu2017b}.
Following the same method, it is easy to obtain the additivity
of coherence weight $C_w$\cite{Bu2017asym} as following,
\begin{eqnarray}
1-C_w(\rho_1\ot \rho_2)
=[1-C_w(\rho_1)]\cdot[1-C_w(\rho_2)].
\end{eqnarray}
Thus, the additivity of
robustness of coherence and
coherence weight are proved here, which will be useful to
the further study on the distribution of coherence in
multipartite systems quantified by robustness of coherence
and coherence weight.

\section{conclusion}
In this work, we have investigated the properties of the incoherent-quantum coherence measures defined by
relative entropy, max-relative entropy and $l_1$ norm on bipartite systems. We also introduce the IQ coherence of formation and assistance
on bipartite systems. And we have found the distribution of coherence of formation $C_f$ and assistance $C_a$ in bipartite systems: the total coherence of formation is lower ( upper ) bounded by the sum of coherence of formation (assistance) in each local subsystem and
entanglement of formation ( assistance ) between subsystems.
Besides, we have obtained the tradeoff relation between
coherence cost and entanglement cost, distillable coherence and distillable entanglement  in bipartite systems.
Moreover, we have  obtained the monogamy relationship of the IQ coherence of formation and assistance in tripartite systems.
Furthermore, the additivity of IQ coherence measures have been discussed.
 These results substantially advance the understanding of the physical laws that governs the distribution of quantum coherence in
bipartite systems and pave the way for the further researches in this direction.

\begin{acknowledgments}
K.F. Bu thanks Prof. Arthur Jaffe and Dr. Zhengwei Liu for their
hospitality in Harvard University and the support of  a grant from the Templeton
Religion Trust.
J.D. Wu is supported by the Natural Science Foundation of China (Grants No. 11171301, No. 10771191, and No. 11571307) and the Doctoral Programs Foundation of the Ministry
of Education of China (Grant No. J20130061).
S.M. Fei is supported by the Natural Science Foundation of China under No. 11675113.

\end{acknowledgments}

\bibliographystyle{apsrev4-1}
 \bibliography{Maxcoh-lit}

\appendix
\section{Properties of $C^{A|B}_f$}\label{apen:cf}
\begin{mproof}[Proof of strong monotonicity under IO on A side for $C^{A|B}_f$]
For any incoherent operation
$\Lambda^A_{IO}$ on A side with the set of Kraus operators $\set{K^A_{\mu}}$,
we need to prove $\sum_{\mu} q_{\mu}C^{A|B}_f(\rho_{AB})(\rho_{\mu})\leq C^{A|B}_f(\rho_{AB})$.
Due to the definition of $C^{A|B}_f$, there exists an optimal pure state decomposition of $\rho_{AB}=\sum_ip_i\proj{\psi_i}_{AB}$ such that
$C^{A|B}_f(\rho_{AB})=\sum p_iC^{A|B}_{r}(\proj{\psi_i}_{AB})$.
Let $\ket{\phi_{i,\mu}}_{AB}=K^A_{\mu}\ket{\psi_i}_{AB}/\sqrt{q_{i,\mu}}$ with $q_{i,\mu}=\Tr{K^A_{\mu}\proj{\psi_i}_{AB}K^{A,\dag}_{\mu}}$. Hence,
\begin{eqnarray*}
C^{A|B}_f(\rho_{AB})&=&\sum_ip_iC^{A|B}_r(\proj{\psi_i}_{AB})\\
&\geq& \sum_ip_i\sum_{\mu} q_{i,u}C^{A|B}_r(\proj{\phi_{i,\mu}})\\
&=& \sum_{\mu}q_{\mu}\sum_i \frac{p_iq_{i,\mu}}{q_{\mu}}C^{A|B}_r(\proj{\phi_{i,\mu}})\\
&\geq&\sum_{\mu}q_{\mu}C^{A|B}_r(\rho_{\mu}),
\end{eqnarray*}
where the first inequality comes from the fact that $C^{A|B}_r$ is strong monotonicity under IO on A side for
$\ket{\psi_i}_{AB}$ and the last inequality comes from the convexity of $C^{A|B}_{r}$.

\end{mproof}

\begin{mproof}[Proof of  monotonicity under CPTP maps on B side for $C^{A|B}_f$]
For CPTP map
$\Lambda^B$ on B side with the set of Kraus operators $\set{K^B_{\mu}}$,
we need to prove $C^{A|B}_f(\rho_{AB})(\Lambda^B(\rho_{\mu}))\leq C^{A|B}_f(\rho_{AB})$.
Due to the definition of $C^{A|B}_f$,
there exists an optimal pure state decomposition of $\rho_{AB}=\sum_ip_i\proj{\psi_i}_{AB}$ such that
$C^{A|B}_f(\rho_{AB})=\sum p_iS(\Delta_A(\proj{\psi_i}_{AB}))=\sum_ip_iS(\Delta_A(\rho^A_i))$ with $\rho^A_i=\Ptr{B}{\proj{\psi_i}_{AB}}$. Hence
$\Lambda^B(\rho_{AB})=\sum_{i,\mu}p_iK^B_{\mu}\proj{\psi_i}_{AB}K^{B,\dag}_{\mu}=\sum_ip_i\sum_{\mu}q_{i,\mu}\proj{\phi_{i,\mu}}_{AB}$
with $\ket{\phi_{i,\mu}}_{AB}=K^A_{\mu}\ket{\psi_i}_{AB}/\sqrt{q_{i,\mu}}$ and $q_{i,\mu}=\Tr{K^A_{\mu}\proj{\psi_i}_{AB}K^{A,\dag}_{\mu}}$. Thus
\begin{eqnarray*}
C^{A|B}_f(\Lambda^B(\rho_{AB}))
&\leq& \sum_{i,\mu}p_iq_{i,\mu}S(\Delta_A(\proj{\phi_{i,\mu}}))\\
&=&\sum_{i,\mu}p_iq_{i,\mu}S(\Delta_A\ot\mathrm{Tr}_B(\proj{\phi_{i,\mu}}))\\
&\leq& \sum_{i}p_iS(\sum_{\mu}q_{i,\mu}\Delta_A\ot\mathrm{Tr}_B(\proj{\phi_{i,\mu}}))\\
&=&\sum_ip_iS(\Delta_A(\rho^A_i))\\
&=&C^{A|B}_f(\rho_{AB}),
\end{eqnarray*}
where the first inequality comes from the definition of $C^{A|B}_f$ and the second inequality comes from the fact the concavity
of von Neumann entropy.
\end{mproof}

\begin{lem}\label{lem:c_fd}
Given a bipartite pure state $\ket{\psi}_{AB}$, it holds that
\begin{eqnarray}
S(\Delta_A(\psi_{AB}))=S(\Delta_A(\Ptr{B}{\psi_{AB}})).
\end{eqnarray}
\end{lem}
\begin{proof}
Since the pure state $\ket{\psi}_{AB}$ can be expressed in the given basis $\set{\ket{i}_A}_i$ of
$\cH_A$ as follows,
\begin{eqnarray*}
\ket{\psi}_{AB}=\sum_i\sqrt{p_i}\ket{i}_A\ket{u_i}_B,
\end{eqnarray*}
with $p_i\geq 0$ and $\sum_ip_i=1$. Then
\begin{eqnarray*}
\Delta_A(\psi_{AB})&=&\sum_ip_i\proj{i}_A\ot \proj{u_i}_B,\\
\Delta_A(\Ptr{B}{\psi_{AB}})&=&\sum_ip_i\proj{i}_A,
\end{eqnarray*}
which implies that $S(\Delta_A(\psi_{AB}))=S(\Delta_A(\Ptr{B}{\psi_{AB}}))=-\sum_ip_i\log p_i$.

\end{proof}

\section{Strong monotonicity under IO on A side }\label{lem:str_mon}
The strong monotonicity under IO on A side can be proved
following the similar method used in \cite{Baumgratz2014}.
For incoherent operation on A side
$\Lambda^A_{IO}(\cdot)=\sum_\mu K^A_\mu(\cdot) K^{A,\dag}_\mu$,
$\rho^{AB}_{\mu}=K^A_n\rho_{AB} K^{A,\dag}_\mu/p_\mu$ and
$p_\mu=\Tr{K^A_\mu\rho_{AB} K^{A,\dag}_\mu}$. Thus
\begin{eqnarray*}
\sum_\mu p_\mu C^{A|B}_{l_1}(\rho^{AB}_{\mu})
&=&\sum_\mu p_\mu\sum_{i\neq j}\norm{\rho^{AB}_{\mu}}_{\mathrm{tr}}\\
&=&\sum_\mu\sum_{i\neq j }\norm{\Innerm{i}{K^A_\mu\rho_{AB}K^{A,\dag}_\mu}{j}}_{\mathrm{tr}}\\
&=&\sum_\mu\sum_{i\neq j}\norm{\sum_{r,s}[K^A_\mu]_{ir}[K^{A,\dag}_\mu]_{sj}\rho^B_{rs}}_\mathrm{{tr}},
\end{eqnarray*}
where $\rho^B_{rs}=\Innerm{r}{\rho^{AB}_{\mu}}{s}_A\in\cB(\cH_B)$.

Since $\Lambda^A$ is incoherent, then
$[K^A_\mu]_{ir}[K^{A,\dag}_\mu]_{rj}=\delta_{ij}$ for any $r$, where
$\delta_{ij}=1$ if $i=j$, otherwise $\delta_{ij}=0$. Therefore
\begin{eqnarray*}
\sum_\mu p_\mu C^{A|B}_{l_1}(\rho^{AB}_{\mu})
&=&\sum_\mu\sum_{i\neq j}\sum_{r\neq s}\norm{\sum_{r,s}[K^A_\mu]_{ir}[K^{A,\dag}_\mu]_{sj}\rho^B_{rs}}_{\mathrm{tr}}\\
&\leq&\sum_\mu\sum_{i\neq j}\norm{[K^A_\mu]_{ir}[K^{A,\dag}_n]_{sj}\rho^B_{rs,\mu}}_{\mathrm{tr}}\\
&=&\sum_{r\neq s}\norm{\rho^B_{rs,\mu}}_{\mathrm{tr}}
\sum_\mu \sum_{i\neq j}|[K^A_n]_{ir}[K^{A,\dag}_\mu]_{sj}|\\
&\leq& \sum_{r\neq s}\norm{\rho^B_{rs}}_{\mathrm{tr}},
\end{eqnarray*}
where the last inequality comes from the
fact that $\sum_\mu\sum_{i\neq j}|[K^A_\mu]_{ir}[K^{A,\dag}_\mu]_{sj}|\leq 1$
given in \cite{Baumgratz2014}. Thus, we obtain the strong monotonicity
of $C^{A|B}_{l_1}$ under IO on A side.

\section{Relation between $l_1$ norm and trace norm}\label{appen:l1_tr}
The trace norm $\norm{\cdot}_{\mathrm{tr}}$ is closely related to the $l_1$ norm $\norm{\cdot}_{l_1}$, for which we have the following relationship,
\begin{lem}\label{lem:11_vs_tr}
Given an operator $P\in\cB(\cH)$ and a fixed reference basis $\set{\ket{i}}_i$ of $\cH$. Then
\begin{eqnarray}
\norm{P}_{\mathrm{tr}}=\min_{U,V}\norm{UPV}_{l_1},
\end{eqnarray}
where the minimization is taken over all the unitaries $U,V$ acting on $\cH$ and $\norm{\cdot}_{l_1}$ is defined by the
the given basis.
\end{lem}
\begin{proof}
First, let us prove that for operator $P\in\cB(\cH)$, $\norm{P}_{\mathrm{tr}}\leq \norm{P}_{l_1}$.
Due to single value decomposition of $P$, there exists two orthonormal basis $\set{\ket{x_i}}_i$ and
$\set{\ket{y_i}}_i$ such that
\begin{eqnarray*}
\norm{P}_{\mathrm{tr}}&=&\sum_i\Innerm{x_i}{P}{y_i}\\
&=&\sum_i\sum_{j,k}\iinner{x_i}{j}\!\Innerm{j}{P}{k}\!\iinner{k}{y_i}\\
&\leq&\sum_{j,k}|\Innerm{j}{P}{k}|\sum_i|\iinner{x_i}{j}\iinner{k}{y_i}|\\
&\leq& \sum_{j,k}|\Innerm{j}{P}{k}|,
\end{eqnarray*}
where the last inequality comes from the fact that
$\sum_i|\iinner{x_i}{j}\iinner{k}{y_i}|\leq (\sum_i|\iinner{x_i}{j}|^2)^{\frac{1}{2}}(\sum_i|\iinner{k}{y_i}|^2)^{\frac{1}{2}}=1$
with $\set{\ket{x_i}}_i$ and  $\set{\ket{y_i}}_i$ being the orthonormal basis.
Thus, $\norm{P}_{\mathrm{tr}}=\norm{UPV}_{\mathrm{tr}}\leq\norm{UPV}_{l_1}$ for any two unitaries.

Besides, there exist unitaries $U$ and $V$ such that
$UPV=\sum_is_i\proj{i}$ with $\set{s_i}_i$ being the single value of
$P$, and thus $\norm{UPV}_{l_1}=\norm{P}_{\mathrm{tr}}$.

\end{proof}

\end{document}